\documentclass[11pt]{article} 

\usepackage{ifthen} 

\usepackage{amssymb,amsmath}
\usepackage{amsthm} 
\usepackage{url}
\usepackage[bookmarksnumbered,colorlinks,linkcolor=blue,citecolor=blue,urlcolor=blue]{hyperref}

\usepackage[usenames,dvipsnames]{color}

\usepackage{algorithm}
\usepackage{algorithmic}

\usepackage[margin=1in]{geometry} 

\DeclareMathOperator{\poly}{poly}
\DeclareMathOperator{\E}{E}
\DeclareMathOperator{\Var}{Var}
\DeclareMathOperator{\lcm}{lcm}

\newtheorem{theorem}{Theorem}
\newtheorem{prop}[theorem]{Proposition}
\newtheorem{lemma}[theorem]{Lemma}

\newcommand{\app}[1]{\hyperref[app:#1]{Appendix~\ref*{app:#1}}}
\renewcommand{\sec}[1]{\hyperref[sec:#1]{Section~\ref*{sec:#1}}}
\newcommand{\thm}[1]{\hyperref[thm:#1]{Theorem~\ref*{thm:#1}}}
\newcommand{\lem}[1]{\hyperref[lem:#1]{Lemma~\ref*{lem:#1}}}
\newcommand{\prp}[1]{\hyperref[prop:#1]{Proposition~\ref*{prop:#1}}}
\newcommand{\alg}[1]{\hyperref[alg:#1]{Algorithm~\ref*{alg:#1}}}
\newcommand{\cor}[1]{\hyperref[cor:#1]{Corollary~\ref*{cor:#1}}}
\newcommand{\eq}[1]{(\ref{eq:#1})}

\newcommand{\set}[1]{\lbrace #1 \rbrace}
\newcommand{\union}{\cup}

\renewcommand{\setminus}{-}

\newcommand{\ket}[1]{| #1 \rangle}

\newcommand{\novert}{*}

\begin{document}

\title{Quantum Property Testing for Bounded-Degree Graphs}
\author{ 
Andris Ambainis\thanks{Faculty of Computing, University of Latvia. Email: ambainis@lu.lv}
\and
Andrew M.\ Childs\thanks{Department of Combinatorics \& Optimization and Institute for Quantum Computing, University of Waterloo. Email: amchilds@uwaterloo.ca}
\and 
Yi-Kai Liu\thanks{Department of Computer Science, University of California, Berkeley. Email: yikailiu@eecs.berkeley.edu}
}
\date{}
\maketitle

\begin{abstract}
We study quantum algorithms for testing bipartiteness and expansion of bounded-degree graphs.  We give quantum algorithms that solve these problems in time $\tilde O(N^{1/3})$, beating the $\Omega(\sqrt{N})$ classical lower bound.  For testing expansion, we also prove an $\tilde\Omega(N^{1/4})$ quantum query lower bound, thus ruling out the possibility of an exponential quantum speedup.  Our quantum algorithms follow from a combination of classical property testing techniques due to Goldreich and Ron, derandomization, and the quantum algorithm for element distinctness.  The quantum lower bound is obtained by the polynomial method, using novel algebraic techniques and combinatorial analysis to accommodate the graph structure.
\end{abstract}


\section{Introduction}

In \emph{property testing}, one is asked to distinguish between objects that satisfy a property $P$ and objects that are far from satisfying $P$.  The goal is to design algorithms that test properties in sublinear or even constant time, without reading the entire input---a task that is nontrivial even for properties that can be computed in polynomial time.  This is motivated by the practical question of how to extract meaningful information from massive data sets that are too large to fit in a single computer's memory and can only be handled in small pieces.

Testing properties of graphs is an interesting special case.\footnote{Here, the graph can be specified by an adjacency matrix (suitable for dense graphs) or by a collection of adjacency lists (for bounded-degree graphs).}  Many graph properties, such as connectivity and planarity, can be tested in constant time, independent of the number of vertices $N$ \cite{GGR98,GR02}.  However, some graph properties are much harder to test.  For bounded-degree graphs in the adjacency-list representation, the best classical algorithms for testing \textit{bipartiteness} \cite{GR99} and \textit{expansion} \cite{GR00,Czumaj07,Kale07,NS07} use $\tilde O(\sqrt{N})$ queries.\footnote{We use tilde notation to suppress logarithmic factors.}  In fact, this is nearly optimal, as there are $\Omega(\sqrt{N})$ query lower bounds for both problems \cite{GR02}.  As a natural extension, we consider whether these problems can be solved more efficiently using \textit{quantum} queries.

There has been some previous work on quantum property testing.  In particular, there are examples of exponential separations between quantum and classical property testing \cite{BFNR08}, and there are quantum algorithms for testing juntas \cite{AS07}, solvability of black-box groups \cite{IL08}, uniformity and orthogonality of distributions \cite{Bravyi09, Chak10}, and certain properties related to the Fourier transform \cite{Aaronson10, Chak10}.
However, aside from concurrent work on testing graph isomorphism \cite{Chak10}, we are not aware of previous work on quantum algorithms for testing properties of graphs.\footnote{Quantum speedups are known for \emph{deciding} certain graph properties, without the promise that the graph either has the property or is far from having it \cite{Durr04,Magniez07,CK11}.  This turns out to be a fairly different setting, and the results there are not directly comparable to ours.}

Here, we give quantum algorithms for testing bipartiteness and expansion of bounded-degree graphs in time only $\tilde O(N^{1/3})$, beating the $\Omega(\sqrt N)$ classical lower bounds \cite{GR02}.  Moreover, we prove that any quantum algorithm for testing expansion must use $\tilde\Omega(N^{1/4})$ queries, showing that quantum computers cannot achieve a superpolynomial speedup for this problem.


Why might quantum computers offer an advantage for testing bipartiteness and expansion?  The classical algorithms for these problems use random walks to explore the graph, so one might hope to do better by using quantum walks, which are a powerful tool for searching graphs \cite{San08}.  In fact, our algorithms use quantum walks indirectly.  The classical algorithm for testing bipartiteness is based on checking whether a pair of short random walks form an odd-length cycle in the graph, thereby certifying non-bipartiteness \cite{GR99}.  The algorithm for testing expansion looks for collisions between the endpoints of short random walks, with a large number of collisions indicating that the walk is not rapidly mixing \cite{GR00}.  In both cases, the property is tested by looking for collisions among a set of $\tilde O(\sqrt{N})$ items.  By using the quantum walk algorithm for element distinctness \cite{Amb07,MNRS07} to look for these collisions, we can solve the problem using $\tilde O(N^{1/3})$ quantum queries.
In addition, we show that the above classical algorithms can be derandomized, using $O(\log N)$-wise independent bits.  This yields quantum algorithms that run in time $\tilde O(N^{1/3})$.  


While we have shown a polynomial quantum speedup, one may ask whether an exponential speedup is possible.  Quantum computers can give at most a polynomial speedup for total functions \cite{BBCMW01}, but this limitation does not apply to property testing (and indeed, examples of exponential speedup are known \cite{BFNR08}).
On the other hand, superpolynomial speedup is impossible for symmetric functions \cite{AA11}, even in the case of partial functions such as those arising in property testing.
It is an interesting question whether exponential speedups are possible for testing \textit{graph} properties, which may have significantly less symmetry.

Here we prove that testing expansion requires $\tilde\Omega(N^{1/4})$ quantum queries, thus ruling out the possibility of an exponential speedup.  We use the polynomial method \cite{BBCMW01}---specifically, a technique of Aaronson based on reduction to a bivariate polynomial \cite{Aar02}.  We define a distribution over $N$-vertex graphs with $\ell$ connected components (and with another parameter $M \approx N$), such that each component is an expander with high probability.  With $\ell=1$ component, such graphs are almost surely expanders, whereas graphs with $\ell\ge2$ components are very far from expanders.  Our main technical contribution is to show that the acceptance probability of any $T$-query quantum algorithm, when presented with this distribution, is well-approximated by a bivariate polynomial in $M$ and $\ell$ of degree $O(T \log T)$.  This requires a somewhat involved calculation of a closed-form expression for the acceptance probability as a function of $M$ and $\ell$, using algebraic techniques and the combinatorics of partitions.  Then it follows by known results on polynomial approximation that $\Omega(N^{1/4}/\log N)$ queries are necessary to test expansion.

This proof may be of independent interest since there are very few techniques available to prove quantum lower bounds for property testing.  In particular, the standard quantum adversary method \cite{Amb02} is subject to a ``property testing barrier'' \cite{HLS07}.  Furthermore, graph structure makes it difficult to apply the polynomial method, so our lower bound for testing expansion requires substantial new machinery.  These techniques may be applicable to other problems with graph structure.  Note also that our approach uses very different techniques from the classical lower bounds for testing bipartiteness and expansion \cite{GR02}.  

We are only aware of a few previous lower bounds for quantum property testing: the result that not all languages can be tested efficiently \cite{BFNR08} (which is nonconstructive, using a counting argument), and lower bounds for testing orthogonality and uniformity of distributions \cite{Bravyi09,Chak10} and for testing graph isomorphism \cite{Chak10} (which follow by reduction from the collision problem).  

Despite this progress, there remain many unanswered questions about quantum testing of graph properties.  So far, we have been unable to prove a superconstant lower bound for testing bipartiteness.  More generally, is there any graph property testing problem that admits an exponential quantum speedup?

\vspace{10pt}

In the remainder of this section, we define the model of quantum property testing.  We use the adjacency-list model for graphs with bounded (i.e., constant) maximum degree $d$.  A graph $G=(V,E)$ is represented by a function $f_G\colon V \times \{1,\ldots,d\} \to V \cup \{\novert\}$, where $f_G(v,i)$ returns the $i^{\rm th}$ neighbor of $v$ in $G$, or $\novert$ if $v$ has fewer than $i$ neighbors.  
A quantum computer is provided with a unitary black box that reversibly computes $f_G$ as $\ket{v,i,z} \mapsto \ket{v,i,z \oplus f_G(v,i)}$.
The query complexity of an algorithm is the number of calls it makes to the black box for $f_G$.  

We say that $G$ is \emph{$\varepsilon$-far} from satisfying a property $P$ if one must change at least $\varepsilon n d$ edges of $G$ in order to satisfy $P$.  We say that an algorithm \emph{$\varepsilon$-tests} $P$ if it accepts graphs that satisfy $P$ with probability at least $2/3$, and rejects graphs that are $\varepsilon$-far from satisfying $P$ with probability at least $2/3$.  (More generally, we may consider algorithms that determine whether a graph satisfies $P$ or is $\varepsilon$-far from satisfying a related property $P'$.)

We say that a graph $G$ is an \emph{$\alpha$-expander} if for every $U \subseteq V$ with $|U| \le |V|/2$, we have $|\partial(U)| \ge \alpha|U|$, where $\partial(U)$ is 
the set of vertices in $V \setminus U$ adjacent to at least one vertex of $U$.


\section{Quantum Algorithms for Bipartiteness and Expansion} 

First, recall the classical algorithm for testing bipartiteness \cite{GR99}.  This algorithm performs $T = \Theta(1/\varepsilon)$ repetitions, where during each repetition it chooses a random starting vertex $s$, then does $K = \sqrt{N} \poly(\frac{\log N}{\varepsilon})$ random walks from $s$, each of length $L = \poly(\frac{\log N}{\varepsilon})$, and looks for ``collisions'' where two walks from $s$ reach the same vertex $v$, one after an even number steps, the other after an odd number of steps.  

We derandomize each of the $T$ repetitions separately.  Each repetition uses $n = O(KL\log d)$ bits of randomness.  We claim that it suffices to use $k$-wise independent random bits for some $k = O(L\log d)$.  To see this, consider the analysis given in \cite{GR99}.  Lemma 4.5 of \cite{GR99} states sufficient conditions for the algorithm to find an odd cycle, and hence reject, with high probability.  The proof considers the random variable $X = \sum_{i<j} \eta_{ij}$, where $\eta_{ij}$ is a Boolean random variable that indicates whether walk $i$ collides with walk $j$ while having different parity.  The probability that $X=0$ is upper bounded using Chebyshev's inequality together with bounds on $\E[X]$ and $\Var[X]$.  Note that $\E[X]$ and $\Var[X]$ are linear and quadratic in the $\eta_{ij}$, respectively, so they only depend on sets of at most $O(L\log d)$ random bits.  Thus they are unchanged by substituting $k$-wise independent random bits for some $k = O(L\log d)$.  This reduces the number of random bits required by the algorithm to $O(k \log n) = O(\poly(\frac{\log N \log d}{\varepsilon}))$.

We then combine this derandomized classical algorithm with Ambainis' quantum algorithm for element distinctness \cite{Amb07,MNRS07,Sze04}.  (For details, see \app{bip}.) 
This shows 
\begin{theorem}
There is a quantum algorithm that always returns ``true'' when $G$ is bipartite, returns ``false'' with constant probability when $G$ is $\varepsilon$-far from bipartite, and runs in time $O(N^{1/3} \poly(\frac{\log N}{\varepsilon}))$.
\end{theorem}

Using similar ideas, we can also give an $\tilde{O}(N^{1/3})$-time quantum algorithm for testing expansion.  We start with the classical algorithm of \cite{GR00}, derandomize it using $k$-wise independent random variables, and apply the quantum algorithm for element distinctness.  There is a slight complication, because we need to count collisions, not just detect them.  However, the number of collisions is small---roughly $O(N^{2\mu})$ where $\mu$ is chosen to be a small constant---so we can count the collisions using brute force.  See \app{exp} 
for details.

\section{Quantum Lower Bound for Testing Expansion}
\label{sec:lb}

\subsection{Overview}
\label{sec:lboverview}

We now turn to lower bounds for testing expansion.  Specifically, we prove

\begin{theorem}
Any quantum algorithm for testing expansion of bounded-degree graphs must use 
$\Omega(N^{1/4}/\log N)$ queries.
\end{theorem}

\begin{proof}
We consider random graphs $G$ on $N$ vertices, sampled from the following distribution $P_{M, l}$ (where $M \ge N$ and $l$ divides $M$):
\begin{enumerate}
\item
We start by constructing a random graph $G'$ on $M$ vertices, as follows:
First, we partition the vertices into $l$ sets $V_1, \ldots, V_l$, 
with each set $V_i$ containing $M/l$ vertices.
Then, on each set $V_i$, we create a random subgraph by randomly choosing $c$ perfect matchings on $V_i$ and taking their union. 
(Here $c$ is some sufficiently large constant.)
\item
We then construct $G$ as follows:
First, we pick a subset of vertices $v_1,\ldots,v_N$ from $G'$.  To pick $v_1$, we choose one of the sets $V_1, \ldots, V_l$ uniformly at random, call it $V_j$, and we let $v_1$ be a random vertex from $V_j$. For each subsequent vertex $v_i$, we again select a set $V_j$ uniformly at random, and choose $v_i$ uniformly at random among those vertices of $V_j$ that were not chosen in the previous steps.
Then we let $G$ be the induced subgraph of $G'$ on $v_1,\ldots,v_N$.
\end{enumerate}

The process above fails if we try to choose more than $M/l$ vertices from the same $V_j$. However, the probability of that happening is small---on average, $N/l$ vertices are chosen in each $V_j$. 
We choose $M= (1+\Theta(N^{-0.1})) N$. 
Then a straightforward application of Chernoff bounds 
implies that the process fails with probability at most $e^{-\Omega(N^{0.55})}$.
For more detail, see 
\sec{fail}.

Note that the resulting graph $G$ has degree at most $c$.  The reason for choosing $G$ as a subgraph of $G'$ (rather than constructing $G$ directly) is that this leads to simpler formulas for the probabilities of certain events, e.g., the probability that vertices $v_1$, $v_2$ and $v_3$ all belong to the same component of $G$ is $1/l^2$.  This seems essential for our use of the polynomial method.

If $l=1$, then this process generates an expander with high probability. 
It is well known \cite{P73,MR} that the graph on $M$
vertices generated by taking $c$ perfect matchings is an expander with high probability.
In \sec{expand}, we show that the subgraph that we 
choose is also an expander. (Informally, the main reason is that only a $\Theta(N^{-1/4})$
fraction of the vertices of $G'$ are not included in $G$. This allows us to carry out the
proof of \cite{P73,MR} without substantial changes.)

If $l=2$, then this process generates a disconnected graph with two connected components, each of size roughly $N/2$. Such a graph is very far from any expander graph---specifically, for any $\alpha'$, it is at least about $(\alpha'/2d)$-far from an $\alpha'$-expander of maximum degree $d$.

Therefore, if a quantum algorithm tests expansion, it must accept a random graph generated according to $P_{M, 1}$ with probability at least $2/3$, and a random graph generated according to $P_{M, 2}$ with probability at most $1/3$. (Graphs drawn from $P_{M, l}$ with $l>2$ must also be accepted with probability at most $1/3$, although this fact is not used in the analysis.)

The strategy of the proof is as follows.  We show that for any quantum algorithm run on a random graph from the distribution $P_{M,l}$, the acceptance probability of the algorithm can be approximated by a bivariate polynomial in $M$ and $l$, where the number of queries used by the algorithm corresponds to the degree of this polynomial.  (This is our main technical contribution.)  We then lower bound the degree of this polynomial.  


In more detail, we will prove the following lemma (see \sec{apx}):
\begin{lemma}
\label{lem:apx}
Let $A$ be a quantum algorithm using $T$ queries.
The acceptance probability of $A$ 
(for the probability distribution $P_{M, l}$) is approximated 
(up to an additive error of $e^{-\Omega(N^{0.55})}$) by a fraction 
$\frac{f(M, l)}{g(M, l)}$, where $f(M, l)$ and $g(M, l)$ are polynomials of degree 
$O(T \log T)$ and $g(M, l)$ is a product of factors $(M-(2k-1)l)$ for $k\in\{1, \ldots, T\}$, with $(M-(2k-1)l)$ occurring at most $2T/k$ times.
\end{lemma}


Now choose $a=1+\Theta(N^{-0.1})$ such that $aN$ is even.
We say that a pair $(M, l)$ is \textit{$\delta$-good} if 
$M \in [ aN - \delta^{3/2}, aN + \delta^{3/2}]$, $l\leq \delta$, and $l$ divides $M$.

We then approximate the fraction $\frac{f(M, l)}{g(M, l)}$ (from \lem{apx}) by $\frac{f(M, l)}{(aN)^{\deg g(M, l)}}$. For each term $M-(2k-1)l$, we first replace it by $M$ and then by $aN$. The first step introduces multiplicative error of
$1-\frac{(2k-1)l}{M} \ge 1- \frac{2kl}{N} \approx e^{-2kl/N}$.
For all terms together, the error introduced in this step is at most
$\prod_{k=1}^T ( e^{-2kl/N} )^{2T/k} = 
e^{-4 T^2 l/N}$.
If $T=O(N^{1/4}/\log N)$ and $l=O(N^{1/2})$, the multiplicative error is
$1-o(1)$.

The second approximation step introduces multiplicative error of
\begin{align*}
(\tfrac{M}{aN})^{O(T \log T)}
\approx ( e^{(M-aN)/aN} )^{O(T \log T)}
\leq ( e^{\delta^{3/2}/aN} )^{O(T \log T)}.
\end{align*}
If $\delta=O(N^{1/2})$ and $T=O({N^{1/4}}/{\log N})$, this can be upper bounded by $1+\epsilon$ for arbitrarily small $\epsilon>0$, by appropriately choosing the big-$O$ constant in $T=O({N^{1/4}}/{\log N})$.


Next, we prove a second lemma, which lower bounds the degree of a bivariate polynomial:
\begin{lemma}
\label{lem:scott}
Let $f(M, l)$ be a polynomial such that $|f(aN, 1)-f(aN, 2)| \geq \epsilon$ for some fixed $\epsilon>0$ and,
for any $\delta$-good $(M, l)$, $|f(M, l)| \le 1$.
Then the degree of $f(M, l)$ is $\Omega(\sqrt{\delta})$. 
\end{lemma}
The proof of this lemma follows the collision lower bounds of Aaronson and Shi 
\cite{Aar02,Shi02} and is included in \sec{as} 
for completeness.


We now set $\delta=\Theta(N^{1/2})$ and apply \lem{scott} to $\frac{f(M, l)}{2(aN)^{\deg g(M, l)}}$.  This is a polynomial in $M$ and $\ell$, because the denominator is a constant.
With $M=aN$, its values at $l=1$ and $l=2$ are bounded away from each other by at least $1/3$ since the algorithm works.
Its values at $\delta$-good pairs $(M,l)$ have magnitude at most $1$ because the acceptance probability of the algorithm is in $[0,1]$, so $|\frac{f(M, l)}{2(aN)^{\deg g(M, l)}}| \le \frac{1}{2} + o(1)$.
Thus we find that the degree of $f(M, l)$ must be $\Omega(N^{1/4})$.
It follows that $T = \Omega({N^{1/4}}/{\log N})$ queries are necessary.
\end{proof}


\subsection{Proof of \lem{apx}}
\label{sec:apx}

Here we assume that the process generating a graph $G$ from the probability
distribution $P_{M, l}$ does not fail. (The effect of this process possibly 
failing is considered in \sec{fail}.) 
The acceptance probability of $A$ is a polynomial $P_A$ of degree 
at most $2T$ in Boolean variables $x_{u, v, j}$, where
$x_{u, v, j}=1$ iff $(u, v)$ is an edge in the $j^{\rm th}$ matching.

$P_A$ is a weighted sum of monomials.  
It suffices to show that the expectation of every such monomial has the rational form described in \lem{apx}.
If this is shown, then $\E[P_A]$ is a sum of such fractions:
$\E[P_A]= \frac{f_1(M, l)}{g_1(M, l)} + \frac{f_2(M, l)}{g_2(M, l)} + \cdots$.
We put these fractions over a common denominator, obtaining $\E[P_A]=\frac{f(M, l)}{g(M, l)}$ where 
$g(M, l)=\lcm(g_1(M, l), g_2(M, l), \ldots)$.
In this common denominator, $(M-(2k-1)l)$ occurs at most $2T/k$ times. Therefore, the degree of $g(M, l)$
is at most
$2T \sum_{k=1}^{2T} \frac{1}{k} = O(T \log T)$.
Similarly, the degree of $f(M,l)$ is at most $O(T\log T) + \deg g(M,l) = O(T \log T)$.

Now consider a particular monomial 
$P=x_{u_1, v_1, j_1} x_{u_2, v_2, j_2} \cdots x_{u_d, v_d, j_d}$, where $d = \deg P$. 
Let $G_P$ be the graph with edges $(u_1, v_1), \ldots, (u_d, v_d)$ (i.e., with the edges relevant to $P$) where the edge $(u_a,v_a)$ comes from the $j_a^{\rm th}$ matching.  
Let $C_1, \ldots, C_k$ be the connected components of $G_P$.  
For each component $C_i$, let $X_i$ be the event that every edge $(u_a,v_a)$ in $C_i$ (viewed as a subgraph of $G_P$) is present in the random graph $G$ as part of the $j_a^{\rm th}$ matching.
We have to find an expression for the expectation
\[
 \E[P]=\Pr[X_1\cap X_2 \cap \ldots \cap X_k] .
\]

We first consider $\Pr[X_i]$.  
Let $v_i$ be the number of vertices in $C_i$, and for each matching $j$, let $d_{i, j}$ be the number of variables $x_{u, v, j}$ in $P$ that have $u, v\in C_i$ and label $j$. Note that 
\begin{equation}
\label{eq:tree} 
d_{i, 1}+d_{i, 2}+\cdots+d_{i, c} \geq v_i -1 
\end{equation}
because a connected graph with $v_i$ vertices must have at least $v_i-1$ edges.
We have
\begin{equation}
\label{eq:xi} 
\Pr[X_i] 
= \frac{1}{l^{v_i-1}} \prod_{j=1}^c \prod_{j'=1}^{d_{i, j}} \frac{1}{{M}/{l}-(2j'-1)} 
= \frac{1}{l^{v_i-1}} \prod_{j=1}^c \prod_{j'=1}^{d_{i, j}} \frac{l}{M-(2j'-1)l} . 
\end{equation}
Here $l^{-(v_i-1)}$ is the probability that all $v_i$ vertices are put into the same set $V_j$ (for some $1\leq j\leq l$) (which is 
a necessary condition for having edges among them), and
$\prod_{j'=1}^{d_{i, j}} \frac{1}{{M}/{l}-(2j'-1)}$
is the probability that $d_{i, j}$ particular edges from the $j^{\rm th}$ matching are present. 
(For the first edge $(u, v)$ in the $j^{\rm th}$ matching, the probability that it is present is $\frac{1}{{M}/{l}-1}$, since $u$ is equally likely to be matched with any of ${M}/{l}$ vertices in $V_j$ except for $u$ itself; for the second edge $(u', v')$ in the $j^{\rm th}$ matching, the probability that it is present is $\frac{1}{{M}/{l}-3}$, since $u'$ can be matched with any of ${M}/{l}$ vertices except $u, v, u'$; and so on.  Note that without loss of generality, we can assume that the edges in $P$ from the $j^{\rm th}$ matching are distinct.  If $P$ contains the same edge twice from the same matching, then we can remove one of the duplicates without changing the value of $P$.)

We can rewrite \eq{xi} as 
$\Pr[X_i] = \frac{1}{l^{v_i-1}} \prod_{j=1}^c R_{d_{i, j}}$, 
where we define 
\begin{equation}
\label{eq:rd} 
R_d = \prod_{j'=1}^{d} \frac{l}{M-(2j'-1)l} . 
\end{equation}
We now extend this to deal with multiple components $C_i$ at once, i.e., we want to evaluate $\Pr[\bigcap_{i\in S} X_i]$, where $S\subseteq \{ 1, \ldots, k\}$.  Let $E_S$ be the event that the vertices in $\bigcup_{i\in S} C_i$ (i.e., in any of the components indicated by $S$) are all put into one set $V_j$.  Then
$\Pr[{\textstyle\bigcap_{i\in S} X_i | E_S}] =  \prod_{j=1}^{c} R_{\sum_{i\in S} d_{i, j}}$. 
The event $E_S$ happens with probability $l^{-(\sum_{i\in S} v_i) +1}$,
since the total number of vertices in $\bigcup_{i\in S} C_i$ is $\sum_{i\in S} v_i$.

Let $L=(S_1, \ldots, S_t)$ be a partition of $\{1, 2, \ldots, k\}$.
We call $S_1, \ldots, S_t$ {\em classes} of the partition $L$. 
We say that $S\in L$ if $S$ is one of $S_1, \ldots, S_t$. 
Let $|L|=t$. We say that $L$ is a refinement of $L'$ (denoted $L< L'$) if $L$ can be obtained
from $L'$ by splitting some of the classes of $L'$ into two or more parts.
We write $L\leq L'$ if $L<L'$ or $L=L'$.
When $L<L'$, let $c_{L, L'}$ be the number of sequences $L=L_0<L_1<\cdots<L_j=L'$, with sequences of even length $j$ counting as $+1$ and sequences of odd length $j$ counting as $-1$.
We define $c_{L,L'} = 1$ when $L=L'$.
We have the following partition identity, which will be useful later; the proof is given in \sec{apxproofdetails}. 

\begin{prop}
\label{prop:part}
Suppose $L''<L$. 
Then $\sum_{L'\colon L''\leq L'\leq L} c_{L', L} = 0$.
\end{prop}

We define the expressions
\begin{align}
\label{eq:fl} 
f_L(M, l)&=
\prod_{S\in L} \prod_{j=1}^c R_{\sum_{i\in S} d_{i, j}} \\
\label{eq:fl1} 
 f'_L(M, l) &= \sum_{L'\colon L'\leq L} c_{L', L} f_{L'}(M, l) .
\end{align}

We can now evaluate $\Pr[X_1\cap X_2 \cap \ldots \cap X_k]$ as follows.  
For any partition $L$ of $\set{1,2,\ldots,k}$, let $E_L$ be the event $\bigcap_{S\in L} E_S$.
Let $E'_L$ be the event that $E_L$ happens but no $E_{L'}$ with $L<L'$ happens
(i.e., $L$ is the least refined partition that describes the event).
Then
\[ \Pr[X_1\cap X_2 \cap \ldots \cap X_k] = 
\sum_L \Pr[E'_L] f_L(M,l) .\]
By inclusion-exclusion,
$\Pr[E'_L] = \sum_{L'\colon L\leq L'} c_{L, L'} \Pr[E_{L'}]$.
Now substitute into the previous equation, reorder the sums, and use the definition of $f'_L(M,l)$:
\[ 
\begin{split}
\Pr[X_1\cap X_2 \cap \ldots \cap X_k] 
 &= \sum_{L'} \Pr[E_{L'}] \sum_{L\colon L\leq L'} c_{L, L'} f_L(M,l) 
  = \sum_L \Pr[E_L] f'_L(M,l).
\end{split}
\]
Note that $\Pr[E_L]=\prod_{S\in L} \Pr[E_S] = \prod_{S\in L} l^{-(\sum_{i\in S} v_i) +1} = l^{-(\sum_{i=1}^k v_i) +|L|}$. 
Thus we have 
\begin{equation}
\label{eq:inc-exc} 
\Pr[X_1\cap X_2 \cap \ldots \cap X_k] = 
\sum_{L} l^{-(\sum_{i=1}^k v_i) +|L|} f'_L(M, l).
\end{equation}

We have now written $\Pr[X_1\cap X_2 \cap \ldots \cap X_k]$ as a sum of rational functions of $M$ and $l$.  We can combine these into a single fraction $\frac{f(M,l)}{g(M,l)}$.  It remains to show that this fraction has the properties claimed in \lem{apx}.

First, we claim that the denominator $g(M,l)$ contains at most $2T/k$ factors of $M-(2k-1)l$.
Observe that each $f_L(M, l)$ is a fraction whose denominator consists of
factors $M-(2k-1)l$.  The number of factors in the denominator is equal to 
the number of variables in the monomial $P$, which is at most $2T$. 
By the form of \eq{rd}, 
for each $M-(2k-1)l$ in the denominator, we also
have $M-l$, $M-3l$, $\ldots$, $M-(2k-3)l$ in the denominator.
Therefore, if we have $t$ factors of $M-(2k-1)l$ in the denominator, then the total
degree of the denominator is at least $tk$. Since $tk\leq 2T$, we have $t\leq {2T}/{k}$.
This statement holds for $f_L(M,l)$ for every $L$.  Thus, when we sum the $f_L(M,l)$ to obtain first $f'_L(M,l)$ and then $\Pr[X_1\cap X_2\cap \ldots \cap X_k]$, and put all the terms over a common denominator $g(M,l)$, this statement also holds for $g(M,l)$.

In $\Pr[X_1\cap X_2 \cap \ldots \cap X_k]$, when we sum the $f'_L(M,l)$ in \eq{inc-exc}, we also have factors of $l^{(\sum_{i=1}^k v_i) -|L|}$ in the denominator. \prp{2} shows that these factors are cancelled out by corresponding factors in the numerator. 
\begin{prop} 
\label{prop:2}
$f'_L(M, l)$ is equal to a fraction  whose denominator is a product of factors $(M-(2k-1)l)$ 
and whose numerator is divisible by $l^{(\sum_{i=1}^k v_i) -|L|}$.
\end{prop}

When we combine the different $f'_L(M,l)$ in \eq{inc-exc} into a single fraction $\frac{f(M,l)}{g(M,l)}$, we see that $f$ and $g$ have the desired form.  Also note that $f$ and $g$ have degree $O(T\log T)$, by repeating the same argument used earlier to combine the different monomials $P$.  This completes the proof of \lem{apx}; it remains to show \prp{2}.

\begin{proof}[Proof of \prp{2}] 
Note that $R_d$ contains an obvious factor of $l^d$.  We define 
\[ R'_d = \frac{R_d}{l^d} = \prod_{j'=1}^{d} \frac{1}{M-(2j'-1)l}  \]
and we redefine $f_L(M,l)$ and $f'_L(M,l)$ (equations \eq{fl} and \eq{fl1}) using $R'_d$ instead of $R_d$.
This removes a factor of $l^d$ from the numerator of $R_d$ and a factor of
$l^{\sum_{i, j} d_{i, j}}$ from the numerator of $f_L(M,l)$.
By equation \eq{tree}, 
this factor is at least $l^{(\sum_i v_i)-k}$. 
Therefore, it remains to show that the numerator of the redefined $f'_{L}(M, l)$ is divisible by $l^{k-|L|}$.

Recall that $f'_L(M, l)$ is a sum of terms $f_{L'}(M, l)$ for all $L'\leq L$.  Let us write each term as $f_{L'}(M, l) = 1/\prod_{k\in K(L')} (M-kl)$, where $K(L')$ is a multiset.  We put these terms over a common denominator $\beta_L(M,l) = \prod_{k\in B(L)} (M-kl)$, where $B(L) \supseteq K(L')$ for all $L'\leq L$.  Then we have 
\[
f_{L'}(M,l) = \frac{\alpha_{L'}(M,l)}{\beta_L(M,l)}, \qquad
\alpha_{L'}(M,l) = \prod_{k\in B(L)-K(L')} (M-kl),
\]
\[
f'_L(M,l) = \frac{\alpha'_L(M,l)}{\beta_L(M,l)}, \qquad
\alpha'_L(M,l) = \sum_{L'\colon L'\leq L} c_{L',L} \alpha_{L'}(M,l).
\]
Let $m = |B(L)|$.  Also, let $\tilde{m} = |K(L')| = \sum_{S\in L'} \sum_{j=1}^c \sum_{i\in S} d_{i,j} = \sum_{i=1}^k \sum_{j=1}^c d_{i,j}$, which is independent of $L'$.  Let $m' = |B(L)-K(L')| = m-\tilde{m}$, which depends on $L$ but not on $L'$.

We want to show that $\alpha'_L(M,l)$ is divisible by $l^{k-|L|}$.  First, we multiply out each term $\alpha_{L'}(M,l)$ to get 
$\alpha_{L'}(M,l) = \sum_{i=0}^{m'} e_i(B(L)-K(L')) M^{m'-i} (-l)^i$,
where $e_i$ is the $i^{\rm th}$ elementary symmetric polynomial (i.e., $e_i(B(L)-K(L'))$ is the sum of all products of $i$ variables chosen without replacement from the multiset $B(L)-K(L')$).  We can then write $\alpha'_L(M,l)$ as
\[
\alpha'_L(M,l) = \sum_{i=0}^{m'} \theta_{L,i} M^{m'-i} (-l)^i, \qquad
\theta_{L,i} = \sum_{L'\colon L'\leq L} c_{L',L} e_i(B(L)-K(L')).
\]
It suffices to show that, for all $0\leq i\leq k-|L|-1$, the coefficient $\theta_{L,i}$ is 0.
Note that if $L$ is the finest possible partition $L_*$, then $|L| = k$ and the above claim is vacuous, so we can assume that $L_* < L$.  Also note that $\theta_{L,0} = 0$ by \prp{part} with $L'' = L_*$, so it suffices to consider $i>0$.

For any set of variables $E$ and any $a\geq 0$, define the power-sum polynomial $T_a(E) = \sum_{k\in E} k^a$.  We can write $e_i(B(L)-K(L'))$ in terms of power sums:
\[
e_i(B(L)-K(L')) = \Lambda_{i,L}[T_a(B(L)-K(L')) \colon a=0,1,2,\ldots,i], 
\]
where $\Lambda_{i,L}$ is a polynomial function of the power sums $T_a(B(L)-K(L'))$ of total degree $i$ in the variables $k\in B(L)-K(L')$.  Note that the polynomial $\Lambda_{i,L}$ only depends on the size of the set $B(L)-K(L')$, hence it only depends on $L$, and not on $L'$.  To simplify things, we can write $T_a(B(L)-K(L')) = T_a(B(L)) - T_a(K(L'))$ and absorb the $T_a(B(L))$ term into the polynomial $\Lambda_{i,L}$ to get a new polynomial $\tilde{\Lambda}_{i,L}$.  Then we have $e_i(B(L)-K(L')) = \tilde{\Lambda}_{i,L}[T_a(K(L')) \colon a=0,1,2,\ldots,i]$, and 
\[
\theta_{L,i} = \sum_{L'\colon L'\leq L} c_{L',L} \tilde{\Lambda}_{i,L}[T_a(K(L')) \colon a=0,1,2,\ldots,i].
\]
It suffices to show that, for all $0\leq i\leq k-|L|-1$, the above sum vanishes term-by-term, i.e., for all sequences $\set{a_j}$ such that $a_j\geq 0$ and $\sum_j a_j \leq i$, we have 
\begin{equation}
\label{eqn-sum-prod-taj}
\sum_{L'\colon L'\leq L} c_{L',L} \prod_j T_{a_j}(K(L')) = 0.
\end{equation}

We have 
$T_a(K(L')) = \sum_{S\in L'} \sum_{j=1}^c T_a(\set{1,3,5,\ldots,2(\sum_{i\in S} d_{i,j})-1})$, by the definition of $K(L')$. 
Note that, for any integer $s$, $T_a(\set{1,3,5,\ldots,2s-1}) = T_a(\set{1,2,3,\ldots,2s}) - 2^a T_a(\set{1,2,3,\ldots,s})$, and by Faulhaber's formula, this equals a polynomial $Q_a(s)$ of degree $a+1$, with rational coefficients and no constant term.  We have $T_a(K(L')) = \sum_{S\in L'} \sum_{j=1}^c Q_a(\sum_{i\in S} d_{i,j})$.  
Let $q_{a,\alpha}$ ($\alpha=1,\ldots,a+1$) be the coefficients of $Q_a$.  Then we can rewrite this as 
\[
T_a(K(L')) = \sum_{\alpha=1}^{a+1} q_{a,\alpha} S_\alpha(L'), \text{ where }
S_\alpha(L') = \sum_{S\in L'} \sum_{j=1}^c \Bigl(\sum_{i\in S} d_{i,j}\Bigr)^\alpha.
\]
It suffices to show that the sum in equation (\ref{eqn-sum-prod-taj}) vanishes term-by-term, i.e., for all $0\leq i\leq k-|L|-1$ and for all sequences $\set{\alpha_j}$ such that $\alpha_j\geq 1$ and $\sum_j (\alpha_j-1) \leq i$, we have 
\begin{equation*}
\sum_{L'\colon L'\leq L} c_{L',L} \prod_j S_{\alpha_j}(L') = 0.
\end{equation*}
This final claim is shown by \prp{final} in \sec{apxproofdetails2}. 
This completes the proof of \prp{2}.
\end{proof}


\noindent \textbf{Acknowledgments.}
We thank the anonymous referees for several helpful comments.
AA was supported by ESF project 1DP/1.1.1.2.0/09/APIA/VIAA/044,
FP7 Marie Curie Grant PIRG02-GA-2007-224886 and
FP7 FET-Open project QCS.
AMC and YKL acknowledge the hospitality of the Kavli Institute for Theoretical Physics, where this research was supported in part by the National Science Foundation under Grant No.\ PHY05-51164.
AMC received support from MITACS, NSERC, QuantumWorks, and the US ARO/DTO.
YKL received support from an NSF postdoctoral fellowship and ARO/NSA.
This work was done in part while YKL was at the Institute for Quantum Information at Caltech.






\ifthenelse{\equal{0}{0}}{}{ 

} 


\appendix


\section{Quantum algorithm for testing bipartiteness}
\label{app:bip}

\subsection{Derandomization}

We recall the classical algorithm for testing bipartiteness \cite{GR99}.  This is based on the fact that a bipartite graph contains no cycles of odd length, whereas if a graph is far from bipartite, then it contains many short odd cycles.  The algorithm tries to find an odd cycle by running several random walks from a common starting vertex $s$ and looking for ``collisions'' where two walks reach the same vertex $v$, one after an even number steps, the other after an odd number of steps.  More precisely, the algorithm performs $T = \Theta(1/\varepsilon)$ repetitions, where during each repetition it chooses a random starting vertex $s$, then does $K = \sqrt{N} \poly(\frac{\log N}{\varepsilon})$ random walks from $s$, each of length $L = \poly(\frac{\log N}{\varepsilon})$.  See \alg{cbip} for a precise description.

\begin{algorithm}[t]
\caption{Testing bipartiteness (classical)}
\label{alg:cbip}
\begin{algorithmic}
\REQUIRE Oracle $f_G$ specifying a graph $G$ with $N$ vertices and max degree $d$;
accuracy parameter $\varepsilon$
\FOR{$\tau=1,\ldots,T$ for some $T = \Theta(1/\varepsilon)$}
  \STATE Pick a random vertex $s$
  \FOR{$i=1,\ldots,K$ where $K = \sqrt{N}\poly(\frac{\log N}{\varepsilon})$}
    \STATE Starting from $s$, take a random walk of length $L = \poly(\frac{\log N}{\varepsilon})$, with steps chosen as follows:
    At vertex $v$, for each adjacent vertex $u$, move to $u$ with probability $\frac{1}{2d}$; stay at $v$ with probability $1-\frac{\deg(v)}{2d}$
    \STATE Let $(w_{ij})_{j=0,1,2,\ldots}$ be the sequence of vertices visited during the walk, omitting consecutive repetitions of the same vertex
    (i.e., when the walk stays at the same vertex for more than one time step,
    only include that vertex once in the sequence, so $w_{ij} \neq w_{i(j+1)}$)
  \ENDFOR
  \IF{$w_{ij} = w_{i'j'}$ for some $i,j,i',j'$, where $j$ is even and $j'$ is odd}
    \RETURN ``false''
  \ENDIF
\ENDFOR
\RETURN ``true''
\end{algorithmic}
\end{algorithm}

This algorithm has the following performance guarantees:

\begin{theorem}[Theorem 2 in \cite{GR99}]
\alg{cbip} always returns ``true'' when $G$ is bipartite and returns ``false'' with probability at least $2/3$ when $G$ is $\varepsilon$-far from bipartite.  The algorithm has running time $\poly(\sqrt{N}\frac{\log N}{\varepsilon})$, and in particular, makes $\poly(\sqrt{N}\frac{\log N}{\varepsilon})$ queries.
\end{theorem}

Our first step is to partially derandomize \alg{cbip} using $k$-wise independent random variables.\footnote{We say that a collection of random variables is \textit{$k$-wise independent} if any subset of at most $k$ of the variables is independent.}  Intuitively, this is possible because the algorithm (and the analysis of its performance) only depend on the behavior of pairs of random walks, which are determined by subsets of $\poly(\log N)$ random bits.  Derandomization reduces the number of random bits from $O(\sqrt{N} \poly(\log N))$ to $\poly(\log N)$, which in turn reduces the running time of our quantum algorithm.

We use the following simple construction for $k$-independent random variables:
\begin{prop}[Proposition 6.5 in \cite{ABI86}]
\label{prop:k-indep}
Suppose $n+1$ is a power of 2 and $k$ is odd, $k \leq n$.  Then there exists a uniform probability space $\Omega = \set{0,1}^m$ where $m = 1+\frac{1}{2}(k-1)\log_2(n+1)$, and there exist $k$-wise independent random variables $\xi_1,\ldots,\xi_n$ over $\Omega$, such that $\Pr[\xi_j=1] = \Pr[\xi_j=0] = \frac{1}{2}$.  

Furthermore, there exists an algorithm that, given $i \in \Omega$ and $1 \leq j \leq n$, computes $\xi_j(i)$ in time $O(k \log n)$.
\end{prop}
Note that even more efficient constructions are possible for random variables that are \textit{almost} $k$-wise independent \cite{AGHP92}, which might facilitate slight improvements to the running times of our algorithms.

We derandomize each of the $T$ repetitions of \alg{cbip} separately.  Each repetition uses $O(KL\log d)$ bits of randomness.\footnote{This involves a minor technical issue: the random walk chooses uniformly among $2d$ outcomes, and when $d$ is not a power of 2, we have to approximate the desired distribution.  This can be handled using standard techniques \cite{ABI86}.}
We claim that it suffices to use $k$-wise independent random bits for some $k = O(L\log d)$.  To see this, consider the analysis of \alg{cbip} in \cite{GR99}.  It is clear that when $G$ is bipartite, the algorithm accepts; the main task is to show that when $G$ is $\varepsilon$-far from bipartite, the algorithm rejects.  The proof establishes the contrapositive: assuming the algorithm accepts with probability at least 1/3, one can construct a bipartition of $G$ with few violating edges, thus showing that $G$ is $\varepsilon$-close to bipartite.

Lemma 4.5 of \cite{GR99} states sufficient conditions for the algorithm to find an odd cycle, and hence reject, with high probability.  The proof considers the random variable $X = \sum_{i<j} \eta_{ij}$, where $\eta_{ij}$ is a Boolean random variable that indicates whether walk $i$ collides with walk $j$ while having different parity.  The probability that $X=0$ is upper bounded using Chebyshev's inequality together with bounds on $\E[X]$ and $\Var[X]$.  In particular, $\E[X]$ and $\Var[X]$ are bounded in terms of quantities that only involve the behavior of a single random walk.  (Likewise, the sufficient conditions in Lemma 4.5 only involve a single random walk.)  Using Lemma 4.5, and the fact that the algorithm accepts, one can deduce properties of the graph $G$.  In the remainder of the proof, these properties are used to construct a bipartition of $G$ with few violating edges, as desired.

Note that $\E[X]$ and $\Var[X]$ are linear and quadratic in the $\eta_{ij}$, respectively, so they only depend on the behavior of sets of at most 4 random walks.  Thus they only depend on sets of at most $O(L\log d)$ random bits, so they are unchanged by substituting $k$-wise independent random bits for some $k = O(L\log d)$.  In particular, letting $\tilde{X}$ denote the derandomized version of $X$, we have $\E[\tilde{X}] = \E[X]$ and $\Var[\tilde{X}] = \Var[X]$.  Thus the probability that $\tilde{X}=0$ is upper bounded using the same argument as above.
It follows that we can substitute $k$-wise independent random variables, constructed using \prp{k-indep}, with $n = O(KL\log d)$ and $k = O(L\log d)$.  This reduces the number of random bits required by the algorithm to $O(k \log n) = O(\poly(\frac{\log N \log d}{\varepsilon}))$.

\subsection{A quantum algorithm}

We now give a quantum algorithm (\alg{qbip}) for testing bipartiteness.  The basic idea is to run several random walks starting from the same vertex $s$ and solve the element distinctness problem to find ``collisions'' between these walks.  
We use the following variant of Ambainis's quantum algorithm for element distinctness.  Let $X$ and $Y$ be finite sets.  Suppose we are given oracle access to a function $f\colon X \rightarrow Y$, and let $R \subseteq Y \times Y$ be a symmetric binary relation that we can compute in time $\poly(\log|Y|)$.  We define a \textit{collision} to be a distinct pair $x,x' \in X$ such that $(f(x),f(x')) \in R$.
The following result gives a quantum algorithm for finding collisions \cite{Amb07,MNRS07,Sze04}:
\begin{theorem}[Special case of Theorem 3 in \cite{MNRS07}]
\label{thm:find-collision}
There is a quantum algorithm (with oracle $f$) that finds a collision (with respect to $R$) with constant probability when a collision exists, always returns ``false'' when there are no collisions, and runs in time $O(|X|^{2/3} \cdot \poly(\log|Y|))$.
\end{theorem}

In our application, each element of $X$ is a sequence of coin tosses; the function $f$ computes the endpoint of the corresponding walk in the graph, together with the number of steps along the way; and the relation $R$ tests whether two walks reach the same vertex, one after an even number of steps, the other after an odd number of steps.  
Since we search for collisions among $\tilde{O}(\sqrt{N})$ elements, we require $\tilde{O}(N^{1/3})$ evaluations of the function $f$.  Moreover, $f$ can be computed using only $\poly(\frac{\log N}{\varepsilon})$ queries, so the quantum algorithm uses $O(N^{1/3} \poly(\frac{\log N}{\varepsilon}))$ queries.

It is now clear why derandomizing the classical algorithm is useful: it allows a concise representation of the elements of $X$.  Rather than enumerating them explicitly, which would take $\tilde{O}(\sqrt{N})$ time, we can describe and manipulate them in time $\poly(\frac{\log N}{\varepsilon})$, so the quantum algorithm runs in time $O(N^{1/3} \poly(\frac{\log N}{\varepsilon}))$.

\begin{algorithm}[t]
\caption{Testing bipartiteness (quantum)}
\label{alg:qbip}
\begin{algorithmic}
\REQUIRE Oracle $f_G$ specifying a graph $G$ with $N$ vertices and max degree $d$;
accuracy parameter $\varepsilon$
\FOR{$\tau=1,\ldots,T$ for some $T = \Theta(1/\varepsilon)$}
  \STATE Pick a random vertex $s$
  \STATE Let $K = \poly(\frac{\log N}{\varepsilon}) \sqrt{N}$, $L = \poly(\frac{\log N}{\varepsilon})$, $n = KL$, and $k = \Theta(L)$
  \STATE Using \prp{k-indep}, construct $k$-wise independent random variables $b_{ij}$ taking values in $\set{0,1,\ldots,2d-1}$
  (for $i = 1,\ldots,K$ and $j = 1,\ldots,L$)
  \STATE Let $X = \set{1,\ldots,K} \times \set{1,\ldots,L}$ and $Y = \set{1,\ldots,N} \times \set{0,1}$
  \STATE Define $f\colon X \rightarrow Y$ as follows:
    Given $(i,j)$,
    run a random walk in $G$, starting at $s$, with random coin flips $(b_{i1},\ldots,b_{ij})$.
    Let $v$ be the endpoint of the walk, and
    let $q$ be the number of steps taken in the graph, not counting steps where the random walk chooses to stay at its current location.
    Return $(v, q \bmod 2)$.
  \STATE Define $R \subseteq Y \times Y$ such that $((v,c),(v',c')) \in R$ iff ($v=v'$ and $c \neq c'$)
  \IF{the algorithm from \thm{find-collision} finds a collision in $R$}
    \RETURN ``false''
  \ENDIF
\ENDFOR
\RETURN ``true''
\end{algorithmic}
\end{algorithm}

Note that the element distinctness algorithm requires that the function $f$ is computed unitarily as $U_f\colon \ket{x} \ket{z} \mapsto \ket{x} \ket{z \oplus f(x)}$.  We have access to the unitary operation $U_{f_G}\colon \ket{v,i} \ket{w} \mapsto \ket{v,i} \ket{w \oplus f_G(v,i)}$ provided by the oracle $f_G$.  Since computing $f$ only requires classical operations and queries to $f_G$, we can perform $U_f$ using reversible classical computation and queries to $U_{f_G}$.

\begin{theorem}
\label{thm:qbip}
\alg{qbip} always returns ``true'' when $G$ is bipartite, returns ``false'' with constant probability when $G$ is $\varepsilon$-far from bipartite, and runs in time $O(N^{1/3} \poly(\frac{\log N}{\varepsilon}))$.
\end{theorem}

\begin{proof}[Proof of \thm{qbip}]
When $G$ is bipartite, it has no odd cycles, so \alg{qbip} never finds a collision.  Thus the algorithm returns ``true.''

When $G$ is $\varepsilon$-far from bipartite, the analysis of \cite{GR99} implies that, with constant probability, one of the sets of random walks sampled by the algorithm contains a collision.  Thus the algorithm returns ``false.''  

For the bound on the running time, note that evaluating the $k$-wise independent random variables $b_{ij}(\omega)$ takes time $O(k \log n) = O(\poly(\frac{\log N}{\varepsilon}))$.  Also, each evaluation of the function $f$ takes time $\poly(\frac{\log N}{\varepsilon})$.  Since $X$ has size $O(\sqrt{N} \poly(\frac{\log N}{\varepsilon}))$, finding a collision takes time $O(N^{1/3} \poly(\frac{\log N}{\varepsilon}))$.
\end{proof}


\section{Quantum algorithm for testing expansion}
\label{app:exp}


\subsection{Derandomization}

We now turn to the problem of testing expansion.
We begin by recalling the classical algorithm for this problem, originally due to \cite{GR00}.  The basic idea is to test how rapidly a random walk from some starting vertex $s$ converges to the uniform distribution.  This can be done by running several random walks starting from $s$ and counting the number of collisions among their endpoints---the number of collisions is smallest when the distribution is uniform.  Here we consider the version of the algorithm that appears in \cite{NS07}.  (See \alg{cexp} for details.)  This algorithm makes $T = \Theta(1/\varepsilon)$ repetitions, and during each repetition, it runs $K = N^{1/2+\mu}$ random walks, each of length $L = (16d^2/\alpha^2) \log N$.  (Here $\varepsilon$, $\mu$, $d$, and $\alpha$ are parameters describing the problem of testing expansion.  They play only a minor role in the present discussion.)

\begin{algorithm}[t]
\caption{Testing expansion (classical)}
\label{alg:cexp}
\begin{algorithmic}
\REQUIRE Oracle $f_G$ specifying a graph $G$ with $N$ vertices and max degree $d$;
accuracy parameter $\varepsilon$; expansion parameter $\alpha$; running time parameter $\mu$
\FOR{$\tau=1,\ldots,T$ for some $T = \Theta(1/\varepsilon)$}
  \STATE Pick a random vertex $s$
  \FOR{$i = 1,\ldots,K$ where $K = N^{1/2+\mu}$}
    \STATE Starting from $s$, take a random walk of length $L = (16d^2/\alpha^2) \log N$
    (the random walk proceeds as follows:  at vertex $v$, for each adjacent vertex $u$, move to $u$ with probability $\frac{1}{2d}$; stay at $v$ with probability $1-\frac{\deg(v)}{2d}$)
    \STATE Let $w_i$ be the endpoint of the walk
  \ENDFOR
  \STATE Let $X$ be the number of pairwise collisions among the vertices $w_1,\ldots,w_K$
  \STATE Let $M = \frac{1}{2} N^{2\mu} + \frac{1}{128} N^{1.75\mu}$
  \IF{$X > M$}
    \RETURN ``false''
  \ENDIF
\ENDFOR
\RETURN ``true''
\end{algorithmic}
\end{algorithm}

This algorithm has the following performance guarantee \cite{NS07, Kale07, Czumaj07}:
\begin{theorem}[Theorem 2.1 in \cite{NS07}]
\label{thm:classical-expansion}
Assume $d \geq 3$, $0<\alpha<1$, and $0<\mu<\frac{1}{4}$.  \alg{cexp} runs in time $O(N^{1/2+\mu} \log N \cdot d^2/\varepsilon\alpha^2)$.  Furthermore, there exists a constant $c>0$ (which depends on $d$) such that, for any $0<\varepsilon<1$:
\begin{enumerate}
\item If $G$ is an $\alpha$-expander, then the algorithm returns ``true'' with probability at least $2/3$.
\item If $G$ is $\varepsilon$-far from any $(c\mu\alpha^2)$-expander of degree at most $d$, then the algorithm returns ``false'' with probability at least $2/3$.
\end{enumerate}
\end{theorem}

We partially derandomize \alg{cexp} using $k$-wise independent random variables, with similar motivation and techniques as for testing bipartiteness.  Note that the algorithm originally requires $KL\log d = O(N^{1/2+\mu} (d^2/\alpha^2) \log N \log d)$ random bits for each of the $T$ repetitions.  As for bipartiteness, we derandomize each of the repetitions independently.  

To explain how the derandomization works, we recall the proof of \thm{classical-expansion} in \cite{NS07}.  Letting $\eta_{ij}$ indicate whether walk $i$ collides with walk $j$, one must show that the random variable $X = \sum_{i<j} \eta_{ij}$ is concentrated around its expectation.  This is done using Chebyshev's inequality together with bounds on $\E[X]$ and $\Var[X]$.  In particular, it is possible to bound $\E[X]$ and $\Var[X]$ in terms of quantities that only depend on the behavior of a single random walk.\footnote{This step occurs in Lemma 3.4 of \cite{NS07}, Lemma 3.1 of \cite{Kale07}, Lemma 4.1 of \cite{Czumaj07}, and Lemma 1 of \cite{GR00}.}

Note that $\E[X]$ and $\Var[X]$ are linear and quadratic in the $\eta_{ij}$, respectively, so they only depend on correlations among up to 4 random walks.  These correlations involve subsets of at most $4L\log d$ random bits.  Thus, we can substitute $k$-wise independent random bits, where $k = 4L\log d$.  Letting $\tilde{X}$ denote the number of collisions in the derandomized algorithm, we have $\E[\tilde{X}] = \E[X]$ and $\Var[\tilde{X}] = \Var[X]$.  Thus the proof goes through just as before:  the bounds on $\E[X]$ and $\Var[X]$ also imply that $\tilde{X}$ is concentrated around its expectation.
Finally, note that the derandomized algorithm requires only
$O\big(k \log(KL\log d)\big) = O\big((d^2/\alpha^2) \log^2 N \log d\big)$
bits of randomness for each of the $T$ repetitions.  We have reduced the number of random bits from $\tilde{O}(N^{1/2+\mu})$ to $O(\log^2 N)$.

\subsection{A quantum algorithm}

We now describe a quantum algorithm for testing expansion.  The basic idea is similar to that for bipartiteness, with one additional detail:  we run several random walks from the same starting vertex, then use a quantum algorithm to \textit{count} the number of collisions among the endpoints of the walks.  More precisely, we determine whether the number of collisions is greater or less than $M$, for small values of $M$.  We do this by using the algorithm of \thm{find-collision} to explicitly find up to $M$ collisions, one at a time; this takes time polynomial in $M$.  See \alg{countcollisions}, whose performance is characterized as follows.


\begin{algorithm}[t]
\caption{Counting collisions (quantum)}
\label{alg:countcollisions}
\begin{algorithmic}
\REQUIRE A set $X$, an oracle $f\colon X \rightarrow Y$, a relation $R \subseteq Y \times Y$, and a number $M$
\STATE Initialize $S = \varnothing$
\FOR{$i=1,\ldots,M$}
  \FOR{$j=1,\ldots,t$ for some $t = \Theta(\log M)$}
    \STATE Run the algorithm of \thm{find-collision} to find some distinct $x,x' \in X$ such that $(f(x),f(x')) \in R$ and $(x,x') \notin S$.
    \IF{the algorithm finds a collision $(x,x')$}
      \STATE set $S = S \cup \set{(x,x'),  (x',x)}$
      \STATE break out of the inner \textbf{for} loop
    \ENDIF
  \ENDFOR
  \IF{the algorithm did not find a collision on any of the $t$ tries}
    \RETURN ``false''
  \ENDIF
\ENDFOR
\RETURN ``true''
\end{algorithmic}
\end{algorithm}

\begin{algorithm}[t]
\caption{Testing expansion (quantum)}
\label{alg:qexp}
\begin{algorithmic}
\REQUIRE Oracle $f_G$ specifying a graph $G$ with $N$ vertices and max degree $d$;
accuracy parameter $\varepsilon$; expansion parameter $\alpha$; running time parameter $\mu$
\FOR{$\tau=1,\ldots,T$ for some $T = \Theta(1/\varepsilon)$}
  \STATE Pick a random vertex $s$
  \STATE Let $K = N^{1/2+\mu}$,
  $L = (16d^2/\alpha^2) \log N$,
  $n = KL$, and
  $k = \Theta(L)$
  \STATE Using \prp{k-indep}, construct probability space $\Omega$ and $k$-wise independent random variables $b_{ij}$ taking values in $\set{0,1,\ldots,2d-1}$ (for $i = 1,\ldots,K$ and $j = 1,\ldots,L$)
  \STATE Choose $\omega \in \Omega$ uniformly at random
  \STATE Let $X = \set{1,\ldots,K}$ and $Y = \set{1,\ldots,N}$
  \STATE Define $f\colon X \rightarrow Y$ as follows:
    Given $i$, return the endpoint of the random walk in $G$
    that starts at $s$ and uses random coin flips
    $(b_{i1}(\omega),\ldots,b_{iL}(\omega))$
  \STATE Define $R \subseteq Y \times Y$ such that $(v,v') \in R$ iff $v=v'$
  \STATE Let $M = \frac{1}{2} N^{2\mu} + \frac{1}{128} N^{(1.75)\mu}$
  \FOR {$\sigma=1,\ldots,t$ for some $t = \Theta(1)$}
    \STATE Run \alg{countcollisions} to test whether there are $M+1$ or more collisions
    \IF{\alg{countcollisions} returns ``true''}
      \RETURN ``false''
    \ENDIF
  \ENDFOR
\ENDFOR
\RETURN ``true''
\end{algorithmic}
\end{algorithm}

\begin{lemma}
\label{lem:count-collisions}
\alg{countcollisions} returns ``true'' with constant probability if there are $M$ or more collisions, always returns ``false'' if there are strictly fewer than $M$ collisions, and runs in time $O(M \log M \cdot |X|^{2/3} \cdot \poly(\log|Y|))$.
\end{lemma}

\begin{proof}[Proof of \lem{count-collisions}]
  Suppose there are $M$ or more collisions.  Then in each iteration $(i=1,\ldots,M)$, there are collisions to be found.  Consider what happens in iteration $i$.  Say that the algorithm from \thm{find-collision} returns ``false'' with probability at most $p$ (some constant).  We run that algorithm $t = \log_{1/p}(3M)$ times.  The probability that it returns ``false'' on every attempt is at most $p^t = \frac{1}{3M}$, so the probability that we return ``false'' during iteration $i$ is at most $\frac{1}{3M}$, and by the union bound, the probability that we return ``false'' is at most $1/3$.

The other claims are easy to see.
\end{proof}

Our quantum algorithm for testing expansion is now straightforward; see \alg{qexp}.  We prove the following:

\begin{theorem}
\label{thm:qexp}
\alg{qexp} runs in time $O(N^{1/3+3\mu} \poly(\log N) \cdot (d^2/\varepsilon\alpha^2) \log(d/\alpha))$.  Furthermore, there exists a constant $c>0$ (which depends on $d$) such that, for any $0<\varepsilon<1$:
\begin{enumerate}
\item If $G$ is an $\alpha$-expander, then the algorithm returns ``true'' with probability at least $2/3$.
\item If $G$ is $\varepsilon$-far from any $(c\mu\alpha^2)$-expander of degree at most $d$, then the algorithm returns ``false'' with probability at least $0.6$.
\end{enumerate}
\end{theorem}

\begin{proof}[Proof of \thm{qexp}]
  Suppose $G$ is an $\alpha$-expander.  Then with probability at least $2/3$, the number of collisions in each of the $T$ repetitions is at most $M$.  When this happens, the collision-counting algorithm always returns ``false,'' so we return ``true.''

Suppose $G$ is $\varepsilon$-far from any $(c\mu\alpha^2)$-expander of degree at most $d$.  Then with probability at most $2/3$, the number of collisions is at most $M+1$ in at least one of the $T$ repetitions.  When this happens, the collision-counting algorithm returns ``true'' with probability at least some constant $p$.  We run the collision-counting algorithm $t$ times, where $t = \log_{1/(1-p)} 10$.  The probability that this returns ``false'' every time is at most $(1-p)^t = 1/10$.  Thus, with probability at least $9/10$, the collision-counting algorithm returns ``true'' at least once, so we return ``false.''

The bound on the running time is straightforward.  In particular, implementing the $k$-wise independent random variables $b_{ij}$ takes time and space $O(k \log n) = O((d^2/\alpha^2) \poly(\log N) \log(d/\alpha))$.  Also, note that the only way we query the graph oracle $f_G$ is to evaluate the function $f$.  Evaluating $f$ requires $L = O((d^2/\alpha^2) \log N)$ queries to $f_G$, and the collision-counting algorithm requires $O(N^{1/3+3\mu} \poly(\log N))$ evaluations of $f$.
\end{proof}


\section{Quantum lower bound for testing expansion}
\label{app:lb}

\subsection{Bounding the failure probability and its impact}
\label{sec:fail}

\begin{lemma}
Let $\frac{M}{N} = 1+\Omega(\frac{1}{N^c})$ and $l\leq N^{1/4}$. 
The probability that the process generating $P_{M, l}$ fails is
at most $e^{-\Omega(N^{0.75-2c})}$.
\end{lemma}

\begin{proof}
Let $X_i$ be the number of vertices chosen from $V_i$. 
$X_i$ is a random variable with expectation $\E[X_i]=N/l$. 
Since $l\leq N^{1/4}$, we have $\E[X_i] \geq N^{0.75}$.
We have to bound 
the probability that $X_i> (1+\epsilon) N/l$ where $\epsilon = \Omega(1/N^c)$.
By standard Chernoff bounds,
\[ \Pr\left[ X_i> (1+\epsilon) \frac{N}{l} \right] \leq e^{- \epsilon^2 \E[X_i]/3} \leq
e^{- N^{0.75-2c}/3} .\]
By the union bound, the probability of the process failing is at most
$l\leq N^{1/4}$ times the probability above.
\end{proof}

Thus the expectations $\E[P]$ for monomials $P$ 
calculated in \sec{apx} are within
a factor $1+e^{-\Omega(N^{0.75-2c})}$ of the correct ones.
To estimate the overall error in the expectation $\E[P_A]$,
we also need bounds on the coefficients of various monomials $P$
in the polynomial $P_A$. (A small error times a large coefficient might
mean a larger error.)  These can be obtained as follows.

\begin{lemma}
Let $P_A$ be the polynomial describing the acceptance probability of
a quantum algorithm. Let $P$ be a monomial of degree $k$. 
Then the absolute value of the coefficient of $P$ in the polynomial $P_A$
is at most $2^k$.
\end{lemma}

\begin{proof}
Let $S$ be the set of variables that appear in $P_A$. For a set of variables $S'$,
let $x(S')$ be the input in which $x_i=1$ for all $i\in S'$ and $x_i=0$ for all
other $i$. By inclusion-exclusion, the coefficient of $P$ is equal to
\begin{equation} 
\label{eq:ies}
\sum_{S'\subseteq S} (-1)^{|S|-|S'|} P_A(x(S')) .
\end{equation}
Since $P_A$ describes the acceptance probability of a quantum algorithm,
$0\leq P_A(x(S'))\leq 1$.
Since the sum \eq{ies} contains $2^k$ terms, its magnitude is at most $2^k$.
\end{proof}

Since our input is described by $O(N^2)$ variables, $P_A$ contains at most
${O(N^2) \choose k} \leq N^{2k}$ monomials $P$ of degree $k$. Therefore, the overall
error introduced by the fact that the process generating the probability distribution 
$P_{M, l}$ may fail is at most
\begin{equation}
\label{eq:last} 
N^{2k} 2^k e^{-\Omega(N^{0.75-2c})} .
\end{equation}
Since $k < N^{1/4}$, we have $N^{2k} 2^k = e^{O(N^{1/2} \log N)}$.
If $c<0.125$, then $N^{1/2} \log N = o(N^{0.75-2c})$ and
the overall error \eq{last} is of the order
$e^{-\Omega(N^{0.75-2c})}$.

\subsection{Expansion properties of subgraphs of unions of random matchings}
\label{sec:expand}

Here we prove that graphs drawn from the distribution $P_{M,1}$ (restricted to the $l=1$ case) are expanders with high probability.
As described in \sec{lboverview}, we consider a random graph $G$ on $N$ vertices obtained in two steps:
\begin{enumerate}
\item
Let $G'=(V', E')$ be a union of $c$ perfect matchings on $M$ vertices;
\item
Let $G=(V, E)$ be the induced subgraph on a random subset of $N$ vertices of $G'$.
\end{enumerate}

We follow the proof that a union of 3 random matchings is an expander with high probability,
as described in \cite{Goldreich}.

\begin{lemma}
\label{lem:expand}
Assume that $M \leq (1+ \frac{1}{N^a})N$ for some 
$a>0$ and that $a, c$ satisfy $c\geq 5$ and $ac>2$. Then there exists $\alpha>0$ such that
a random graph $G$ generated according to the probability distribution $P_{M,1}$
is an $\alpha$-expander with probability $1-o(1)$. 
\end{lemma}

\begin{proof}
Let $E_{i, j}$ be the event that in graph $G$, for two sets of vertices $U_1$, $U_2$ with
$|U_1|=i$, $|U_2|=j$, $U_1\cap U_2=\varnothing$, all the neighbors of vertices
$v\in U_1$ belong to $U_1\union U_2$.
By the union bound, the probability that $G$ is not an $\alpha$-expander is upper bounded by 
\begin{equation}
\label{eq:sum-e} 
\sum_{i=1}^{N/2} {N \choose i} {N \choose \alpha i} \Pr [E_{i, \alpha i}] .
\end{equation}

We claim the following:
\begin{prop}
  \label{prop:neighborevent}
\[ \Pr[E_{i, \alpha i}] \leq \left( \frac{1}{N^a} + \frac{(1+\alpha) i}{N} \right)^{\frac{ci}{2}} .\]
\end{prop}

Using this claim, we can upper bound \eq{sum-e} by
\[ \sum_{i=1}^{N/2} {N \choose i} {N \choose \alpha i} 
\left( \frac{1}{N^a} + \frac{(1+\alpha) i}{N} \right)^{\frac{ci}{2}} .\]
If $(1+\alpha) i \leq 3N^{1-a}$, we can upper bound the $i^{\rm th}$ term of this sum by
\[ N^{(1+\alpha) i} 
\left( \frac{4}{N^a}  \right)^{\frac{ci}{2}} =
\frac{2^{ci}}{N^{(\frac{ac}{2}-1-\alpha) i}} . \]
If $ac>2(1+\alpha)$, the sum of all such terms is $o(1)$ because they form a 
geometric progression with common ratio ${2^c}/{N^{ac/2-1-\alpha}}=o(1)$.

Terms with $(1+\alpha) i \leq 3N^{1-a}$ are upper bounded by
\begin{align*} \left( \frac{e N}{i} \right)^i 
\left( \frac{e N}{\alpha i} \right)^{\alpha i} 
\left(  \frac{4(1+\alpha) i}{3N} \right)^{\frac{ci}{2}}
&= 
\left( \frac{i}{N} \right)^{(c/2-1-\alpha)i} \frac{e^{(1+\alpha)i}}{\alpha^{\alpha i}}
\left( \frac{4 (1+\alpha)}{3} \right)^{ci/2} \\
&\leq \left( \frac{1}{2} \right)^{(c/2-1-\alpha)i}
\frac{e^{(1+\alpha)i}}{\alpha^{\alpha i}} 
\left( \frac{4 (1+\alpha)}{3} \right)^{ci/2}.
\end{align*}
If $c\geq 5$ and $\alpha$ is sufficiently small, this is equal to 
$C^i$ with $C=1-\Omega(1)$.  The sum of $C^i$ over all 
$i$ such that $(1+\alpha) i \leq 3N^{1-a}$ is $C^{-\Omega(N^{1-a})}=o(1)$.
\end{proof}

To establish \lem{expand}, it remains to prove \prp{neighborevent}.

\begin{proof}[Proof of \prp{neighborevent}]
To show this, we first observe that $E_{i, \alpha i}$ is equivalent to all neighbors of
vertices $v\in U_1$ in graph $G'$ belonging to $(V' \setminus V) \cup U_1 \cup U_2$.
Let $m= |(V' \setminus V) \cup U_1\cup U_2|$. Then $m\leq \frac{N}{N^a}+(1+\alpha)i$ 
because there are $M-N\leq \frac{N}{N^a}$ 
vertices in $V' \setminus V$ and $(1+\alpha)i$ vertices in $U_1\cup U_2$.
We consider one of $c$ matchings.
The probability that the first vertex $v\in U_i$ is matched to a vertex in 
$(V' \setminus V) \cup U_1\cup U_2$ is equal to $\frac{m-1}{M-1}$. The probability that the 
next vertex $v\in U_1$ is matched to a vertex in $(V' \setminus V) \cup U_1\cup U_2$ is equal to $\frac{m-3}{M-3}$, and so on. Since $|U_1|=i$, there must be at least $i/2$ edges incident 
to a vertex $v\in U_1$. Therefore, the probability that all vertices $v\in U_1$ are matched
to vertices in $(V'\setminus V) \cup U_1\cup U_2$ is
\[ \frac{m-1}{M-1} \frac{m-3}{M-3} \cdots \frac{m-i+1}{M-i+1} \le
\left( \frac{m}{M} \right)^{i/2} .\]
The probability that this happens for all $c$ matchings is 
at most $(\frac{m}{M})^{ci/2}$.
Since $M>N$, we have $\frac{m}{M} < \frac{1}{N^a} + \frac{(1+\alpha) i}{N}$.
The desired result follows.
\end{proof}

\subsection{Lower bound on polynomial degree}
\label{sec:as}

In this section, we prove \lem{scott}. We do not believe that this result is new, but it does not appear in exactly this form in either \cite{Aar02} 
or \cite{Shi02}. Aaronson's argument in \cite{Aar02} is weaker in one place,
while Shi \cite{Shi02} proves a stronger lower bound of $\Omega(n^{1/3})$ 
for the collision problem using a different reduction
to polynomial approximation (which does not seem to be applicable to our
graph problem). Thus, we include the proof of \lem{scott}
for completeness.

The proof uses the following result of Paturi \cite{Pat92}.

\begin{theorem}
\label{thm:pat}
Let $g(x)$ be a polynomial such that $|g(x)|\leq 1$ for all integers $x \in [A,B]$ and
$|g(\zeta)-g(\lfloor \zeta\rfloor)|\geq c$ for some constant $c>0$ and some $\zeta \in [A,B]$.
Then the degree of $g(x)$ is
\[ \Omega\left( \sqrt{(\zeta-A+1)(B-\zeta+1)} \right) .\]
\end{theorem}

\begin{proof}[Proof of \lem{scott}]

We consider the behavior of $f(M,l)$ when we fix $M=aN$.  Similarly to \cite{Aar02}, we consider two cases:

{\bf Case 1.}
Suppose $|f(aN, l)| \le \frac{4}{3}$ for all $l\in\{1, \ldots, \delta\}$.
Then $g(l)=\frac{3}{4} f(aN, l)$ is a polynomial
satisfying $|g(l)| \le 1$ for all $l\in\{1, \ldots, \delta\}$.
We have $|g(1)-g(2)| \geq \frac{3}{4} \epsilon$, so
by \thm{pat}, the degree of $g$ is $\Omega(\sqrt{\delta})$.

{\bf Case 2.}
Suppose $|f(aN, l)|>\frac{4}{3}$ for some $l\in\{1, \ldots, \delta\}$.
Fix this value of $l$, and let $M_0$ be the smallest value for which $(M_0, l)$ is $\delta$-good.
Let $M_1$ be the largest value for which $(M_1, l)$ is $\delta$-good.
We define $g(x)=f(M_0+xl, l)$. Then $|g(x)| \le 1$ for all 
$x\in\{0, 1, \ldots, m\}$ where $m=\frac{M_1-M_0}{l}$.
Also, $|g(x_0)|>\frac{4}{3}$ for
$x_0=\frac{aN-M_0}{l}$.

We have $M_0 = l \lceil \frac{1}{l} (aN-\delta^{3/2}) \rceil <
aN-\delta^{3/2} + l$ and $M_1 = l \lfloor \frac{1}{l} (aN+\delta^{3/2}) \rfloor > aN+\delta^{3/2} - l$.
Therefore, $m> \frac{2 \delta^{3/2}}{l} - 2$. 
Since $l\leq \delta$, this means that $m\geq 2\sqrt{\delta}-2$.
Also, the above bounds on $M_0$ and $M_1$ imply that 
$\frac{M_0+M_1}{2} \in [aN-\frac{l}{2}, aN+\frac{l}{2}]$.
Therefore, $aN \in [\frac{M_0+M_1-l}{2}, \frac{M_0+M_1+l}{2}]$, 
so $x_0 \in [\frac{m-1}{2}, \frac{m+1}{2}]$.
By \thm{pat}, the degree of $g$ must be $\Omega(m)=\Omega(\sqrt{\delta})$. 
\end{proof}

\subsection{A partition identity}
\label{sec:apxproofdetails}

\begin{proof}[Proof of \prp{part}]
We want to establish the identity
\begin{equation}
\label{eq:part} 
\sum_{L'\colon L''\leq L'\leq L} c_{L', L} = 0
\end{equation}
for any $L,L''$ with $L'' < L$.  We prove the claim by induction on $|L''|-|L|$.

If $|L''|-|L|=1$, then the sum \eq{part} is just $1+c_{L'', L}$.
Since the only way to obtain $L''$ by successive refinements of $L$ consists
of one step $L''<L$, we have $c_{L'', L}=-1$ and $1+c_{L'',L}=0$.

For the inductive case, we can express
\begin{equation}
\label{eq:recursive} 
 c_{L', L}= - \sum_{L_1\colon L'\le L_1<L} c_{L', L_1} 
\end{equation}
where the term with $L_1=L'$ counts the path $L' < L$ and a general term counts the paths $L' < \cdots < L_1 < L$.
If we expand each $c_{L', L}$ on the left hand side of \eq{part}
using \eq{recursive}, we get
\begin{align*}
-\sum_{L'\colon L''\leq L'\leq L} \left( \sum_{L_1\colon L'\leq L_1<L} c_{L', L_1} \right)
&= -\sum_{L',L_1\colon L''\leq L' \leq L_1 < L} c_{L', L_1} \\
&= -\sum_{L_1\colon L''\leq L_1<L} \left( \sum_{L'\colon L''\leq L'\leq L_1} c_{L', L_1} \right).
\end{align*}
Each of the terms in brackets on the right hand side is 0 by the inductive assumption.
\end{proof}

\subsection{Proof of \prp{2}: The final step}
\label{sec:apxproofdetails2}

To complete the proof of \prp{2}, it remains to show the following.
\begin{prop}
\label{prop:final}
Let $F(L')=S_{\alpha_1}(L')\ldots S_{\alpha_m}(L')$ where $\alpha_j\geq 1$ and $\sum_j (\alpha_j-1)\leq k-|L|-1$. Then
\begin{equation}
\label{eq:sum-fl} 
\sum_{L'\colon L'\leq L} c_{L', L} F(L') = 0.
\end{equation}
\end{prop}  

\begin{proof}
First, observe that for $\alpha_i=1$, $S_1(L')=\sum_{j=1}^c \sum_{i=1}^k d_{i, j}$, which is 
independent of the partition $L$.
Therefore, $S_1(L')$ is just a multiplicative constant, and it suffices to prove the claim assuming $\alpha_i\geq 2$ for all $i$.

We expand each $F(L')$ in \eq{sum-fl} into a linear combination of terms, where each term is a product of $d_{i, j}$s. Consider one such term.
We can write it as 
\begin{equation}
\label{eq:term} 
\prod_{l=1}^m \prod_{o=1}^{\alpha_l} d_{i_{l, o}, j_{l, o}} 
\end{equation}
where $\prod_{o=1}^{\alpha_l} d_{i_{l, o}, j_{l, o}}$ is the part that comes from expanding
$S_{\alpha_l}(L')$.

We would like to show that the sum of all coefficients of \eq{term} in the expansion
of \eq{sum-fl} is 0.
In order for a term \eq{term} to appear in the expansion of $F(L')$, for each $l$, 
$I_l= \{i_{l, 1}, \ldots, i_{l, \alpha_l}\}$ must be contained 
in one class $S$ of the partition $L'$. 

We consider two cases. If, for some $i$,
$I_l$ is not contained in one class $S$ of the partition $L$,
then $I_l$ is also not contained in one class of any other $L'$,
because all the $L'$s in \eq{sum-fl} are refinements of $L$. Then the term 
\eq{term} does not appear in the expansion of any $F(L')$.

Therefore, we can restrict to terms for which each 
$I_l$ is contained in one class $S$ of the partition $L$.
We consider a partition $L''$ defined as follows.
Let $G$ be a graph with vertex set 
$\{1, \ldots, k\}$ and edges from $i_{l, 1}$ to $i_{l, 2}$,
$\ldots$, $i_{l, \alpha_l}$, for each $l$. 
The classes of $L''$ are the connected components of $G$.
In other words, $L''$ is the finest partition such that, for all $l$, all elements of $I_l$ are in the same class.

If a term \eq{term} appears in the expansion of $F(L')$, then
$L''\leq L'$. 
Therefore, the sum of all the coefficients of \eq{term} 
is
\[ \sum_{L'\colon L''\leq L'\leq L} c_{L', L} .\]
We claim that $L'' < L$; then this sum is 0 by \prp{part}.
Clearly, $L''\leq L$, so it remains to prove that $L\neq L''$.

Observe that the graph $G$ used to define $L''$ has
$\mu = \sum_{j} (\alpha_j-1)$ edges. Therefore, it has at least 
$k-\mu$ connected components, i.e., $|L''|\geq k-\mu$. Since
$\mu\leq k-|L|-1$, we have $|L|\leq k-\mu-1$, so $L\neq L''$.
\end{proof}

\end{document}